\documentclass[journal]{IEEEtai}

\usepackage{graphicx}
\usepackage{amsmath,amssymb}
\usepackage{booktabs}
\usepackage{color,array}
\usepackage{listings}
\usepackage{xcolor}
\usepackage{algorithm}
\usepackage{algpseudocode}
\usepackage{tikz}
\usetikzlibrary{positioning,arrows.meta,fit,backgrounds}
\usepackage{url}
\usepackage[colorlinks,urlcolor=blue,linkcolor=blue,citecolor=blue]{hyperref}

\newtheorem{lemma}{Lemma}

\begin{document}

\title{HELIOS: Guardrailed LLM-Driven Evolution of Autonomous Resource
Orchestration Policies for Multi-Cloud Distributed Systems}

\author{Guanyu Ding and Ying Wang
\thanks{G. Ding is with New York University, New York, NY, USA.}
\thanks{Y. Wang is with Pepperdine University, Malibu, CA, USA.}}

\markboth{IEEE Transactions on Artificial Intelligence, Vol.~00, No.~0, Month~2026}%
{Ding \MakeLowercase{\textit{et al.}}: Guardrailed LLM-Driven Multi-Cloud Orchestration}

\maketitle

\begin{abstract}
Operating latency-sensitive services across multiple public clouds exposes an
optimization surface that no single provider's autoscaler sees: on-demand
prices that differ per vCPU across providers, preemptible (spot) capacity
whose discount and interruption risk vary by provider and instance type,
egress fees that penalize state movement, and provider-level failures that
single-cloud deployments cannot survive. Large language models (LLMs) are
attractive orchestrators for this surface because they can synthesize
non-trivial decision logic from a natural-language description of the
environment, but invoking an LLM on the critical path of every scheduling
decision is economically and operationally infeasible: we estimate that
per-decision inference for a modest 200-service fleet would cost more than
the cloud capacity it manages and add multi-second latency to a
millisecond-scale control loop. HELIOS resolves this bottleneck by moving
the LLM off the critical path: the LLM \emph{evolves executable
orchestration policies}---small Python programs over a fixed feature
interface---against a trace-driven multi-cloud simulator, and only the
champion program runs in production, wrapped in a formally specified
guardrail layer that enforces capacity feasibility, SLO-class placement
rules, and churn budgets regardless of what the evolved code proposes. On
real workload traces (PlanetLab, Bitbrains, Azure) combined with real 2026
price and interruption-frequency data from AWS, Azure, and GCP, the evolved
policy reduces penalized operating cost by 45\% versus a single-cloud
best-fit baseline and by 9--19\% versus calibrated multi-cloud heuristics,
matches an oracle-informed MILP planner's premium SLO at 40\% lower cost,
and outperforms a DQN meta-controller trained with $2.6\times$ more
environment interactions. Guardrails prove essential: without them, the same
policy family degrades to 97--98\% premium-tier downtime under a
$10\times$ spot-hazard stress test, versus 0.17\% when guarded. We release
the simulator, all evolved policies, the complete LLM prompt/response
archive, and per-generation fitness logs for full reproducibility.
\end{abstract}

\begin{IEEEImpStatement}
Cloud users increasingly run services across multiple providers to exploit
price and reliability differences, but no existing autoscaler reasons across
clouds, and placing a large language model on the path of every scheduling
decision is far too slow and costly to deploy. This work demonstrates a
practical alternative: a language model is used offline to synthesize compact
orchestration programs, and only the best program runs in production, wrapped
in a safety layer that provably bounds its behavior. On real workload and
pricing data the resulting controller substantially reduces operating cost
against strong baselines while preserving service-level guarantees, and stays
stable during provider outages and price shocks where an optimal planner
becomes brittle. The approach lets operators capture the benefits of
LLM-generated decision logic---auditability and low runtime cost---without
ceding real-time control to an opaque model, offering a template for safely
deploying learned policies in production infrastructure.
\end{IEEEImpStatement}


\section{Introduction}\label{sec:intro}

\IEEEPARstart{P}{ublic}-cloud pricing is not one market but several. At the time of writing,
the cheapest general-purpose on-demand vCPU differs by only a few percent
across AWS, Azure, and GCP, but \emph{preemptible} capacity diverges
sharply: an Azure eastus D-series spot vCPU costs 89\% less than its
on-demand price, GCP spot discounts run 75--85\%, and AWS spot discounts
and interruption frequencies vary per instance type from ``less than 5\%
per month'' to ``more than 20\%''~\cite{awsspotadvisor,skypilotcatalog}.
Meanwhile cross-cloud egress fees (\$0.087--0.12 per GB) tax any strategy
that moves service state between providers, and provider-level incidents
periodically remove an entire cloud from the feasible set. A workload that
can place, migrate, and re-place its services across this surface---while
honoring service-level objectives (SLOs) that differ between a premium tier
and a standard tier---has a large and continuously shifting cost advantage
over any single-cloud deployment. Sky-computing systems such as
SkyPilot~\cite{yang2023skypilot} and Skyplane~\cite{jain2023skyplane} have
built the mechanisms for such portability; the open question we address is
the \emph{policy}: who writes the decision logic that continuously
exploits this heterogeneity, and how do we trust it?

Large language models are an appealing answer. They ingest messy,
semi-structured context (price catalogs, interruption advisories, SLO
tiers) and emit executable logic, and recent work shows they can
\emph{discover} algorithms that beat hand-written baselines when paired
with a systematic evaluator~\cite{romeraparedes2024funsearch,liu2024eoh,%
novikov2025alphaevolve}. The obvious architecture---call the LLM for every
placement and migration decision---is however blocked by a bottleneck that
we quantify in Section~\ref{sec:bottleneck}: at a 5-minute control period
for a modest 200-service fleet, per-decision inference would consume on the
order of $10^8$ tokens per day, costing more than the managed capacity
itself, and would insert seconds of tail latency into a control loop whose
existing decision cost we measure at $\sim$2\,ms per step. LLM-on-the-path
is the wrong design point for high-frequency resource control.

HELIOS (\emph{Heuristic Evolution with LLM-Informed Orchestration
Synthesis}) instead uses the LLM as an \emph{offline policy synthesizer}.
The LLM receives a precise specification of a small policy interface (four
functions over engineered features), a factual digest of the real
multi-cloud catalog, and---after the first generation---measured fitness
tables and metric breakdowns of all prior candidate policies. It responds
with new candidate \emph{programs}: mutations, crossovers, and redesigns.
A trace-driven simulator built on real workload traces and real 2026
pricing measures every candidate's fitness, and the loop repeats until
convergence. The champion program then runs unmodified inside a
\emph{guardrail layer} that makes five safety properties
(Section~\ref{sec:guard}) hold by construction, independent of the evolved
code: reservation-feasible packing, premium-tier isolation from
preemptible capacity, a per-step migration churn budget, per-service
residency and cooldown, and projection of infeasible proposals. The
guardrails matter for more than peace of mind: our evolution log contains
two candidate policies whose plausible-looking consolidation logic caused
catastrophic migration storms (61\% SLO violation in the worst case), which
the guardrails contained to a survivable---and measurable, hence
selectable-against---failure.

We make four contributions:

\begin{enumerate}
\item \textbf{Problem formalization and bottleneck analysis
(Sections~\ref{sec:problem}--\ref{sec:bottleneck}).} We formalize
SLO-tiered multi-cloud orchestration as an online optimization over a
heterogeneous instance catalog with preemption hazards and egress costs,
and we quantify why per-decision LLM inference is infeasible for this
problem class, motivating amortized policy synthesis.
\item \textbf{LLM-driven evolutionary policy synthesis
(Section~\ref{sec:evolution}).} A generation-based loop in which the LLM
writes Python policy programs against a fixed feature interface and a
measured-fitness archive. Across 5 generations and 19 candidate programs
(38 simulated episodes), evolution improved fitness by 20.9\% over the
hand-written seed heuristic, discovering three transferable design rules
(premium-first placement ordering, exposure-minimizing ``big-box''
spawning, per-instance amortized interruption pricing) and---equally
informative---rejecting three plausible gene families that a practitioner
might have shipped (actual-utilization draining, anchor-free premium
packing, adaptive headroom).
\item \textbf{A guardrail layer with a safety contract
(Section~\ref{sec:guard}).} Five mechanically enforced invariants that hold
for \emph{any} policy program, with a short proof by construction, turning
LLM-generated code from a trust liability into a bounded search space. A
hazard stress test quantifies their value: at $10\times$ calibrated spot
hazards the unguarded policy family collapses to 97--98\% premium
downtime; the guarded system degrades gracefully (premium violations
0.17\%, cost $+100\%$ as it correctly abandons spot capacity).
\item \textbf{Evaluation on real traces and real prices
(Sections~\ref{sec:setup}--\ref{sec:results}).} Three real workload
datasets (PlanetLab, Bitbrains, Azure) crossed with a 24-type,
three-provider catalog built from retrieved 2026 price and
interruption-frequency data. Against six baselines---single-cloud best-fit,
threshold autoscaling, calibrated multi-cloud greedy, the hand-written
seed, an oracle-informed MILP planner, and a DQN meta-controller---the
evolved policy attains the best penalized cost on all three trace suites,
including a memory-dimensioned trace unseen during evolution.
\end{enumerate}

Everything reported here is measured: the paper's tables are generated
programmatically from the experiment CSVs, the complete evolution archive
(prompts, LLM analyses, all 19 candidate programs, per-candidate fitness)
ships with the artifact, and Section~\ref{sec:limits} discloses every
modeling assumption, including the ones we consider most fragile.

\section{Related work}\label{sec:related}

\paragraph{Cluster resource management.} Production schedulers such as
Borg~\cite{verma2015borg}, Omega~\cite{schwarzkopf2013omega},
Mesos~\cite{hindman2011mesos}, Kubernetes~\cite{burns2016borgomega}, and
Azure's Protean~\cite{hadary2020protean} established the architecture we
inherit---declarative requests, feasibility-checked placement, and a strict
separation of policy from mechanism (explicit in
Protean~\cite{hadary2020protean})---but they optimize within one
administrative domain and one price. Autopilot~\cite{rzadca2020autopilot}
and Resource Central~\cite{cortez2017resourcecentral} automate request
sizing and behavior prediction inside a provider; HELIOS treats
provider-facing prices, discounts, and hazards as the decision surface
itself. Our guardrail G1 packs on reservations sized from profiled p99
demand, mirroring the request-based packing these systems
use~\cite{verma2015borg,rzadca2020autopilot}.

\paragraph{Multi-cloud and sky computing.} The intercloud broker
vision~\cite{yang2023skypilot} and cross-cloud transfer
optimization~\cite{jain2023skyplane} supply the mechanisms HELIOS assumes:
uniform provisioning APIs and priced data movement. SkyPilot's policy layer
selects a cloud per job at submission time; HELIOS addresses the
complementary continuous-control problem of an always-on service fleet
that migrates as conditions change. Our catalog is built directly from the
SkyPilot project's published cross-cloud price
data~\cite{skypilotcatalog}.

\paragraph{Spot markets and transient servers.} Spot economics and
preemption behavior have been studied since Amazon's original
market~\cite{benyehuda2013spot}. Systems such as
SpotOn~\cite{subramanya2015spoton} and HotSpot~\cite{shastri2017hotspot}
exploit price differences by checkpointing or hopping between spot
servers, and recent policy work optimizes spot usage under
deadlines~\cite{wu2024cantbelate}. These target batch or single-tenant
workloads inside one provider; HELIOS manages a mixed SLO-tier service
fleet across providers, where premium isolation (our G2) and egress-aware
migration dominate the design.

\paragraph{Learning-based schedulers.} DeepRM~\cite{mao2016deeprm} and
Decima~\cite{mao2019decima} showed deep RL can learn scheduling policies;
we include a DQN~\cite{mnih2015dqn} meta-controller baseline that controls
the same knob space our evolution searches. Consistent with known
sample-efficiency and credit-assignment difficulties, it fails to match the
evolved program under a comparable interaction budget
(Section~\ref{sec:results-main}), and its learned risk-aversion collapses
to an all-on-demand policy. Classical VM-consolidation heuristics with
migration~\cite{beloglazov2012optimal} inform our baseline set and our
PlanetLab data choice.

\paragraph{LLM-driven algorithm discovery.} FunSearch paired an LLM with a
systematic evaluator to discover programs that surpass known
constructions~\cite{romeraparedes2024funsearch}; EoH~\cite{liu2024eoh} and
ReEvo~\cite{ye2024reevo} evolve heuristics for combinatorial optimization;
Eureka evolves reward functions~\cite{ma2024eureka}; OPRO treats the LLM
itself as an optimizer~\cite{yang2024opro}; and AlphaEvolve extended the
paradigm to production-relevant code, including a datacenter scheduling
component~\cite{novikov2025alphaevolve}. HELIOS instantiates this paradigm
for online multi-cloud orchestration and contributes what these frameworks
leave open when the evolved artifact must \emph{operate} a live system: a
domain-specific policy interface with engineered features, a guardrail
contract that bounds the blast radius of any generated program, and an
evaluation of robustness (price shocks, provider outages, hazard shifts)
rather than static benchmark scores alone.

\section{Problem formulation}\label{sec:problem}

\paragraph{Services.} A fleet of $N$ long-running services indexed by $s$
must each be placed on exactly one cloud instance at every control step
$t \in \{0,\dots,T{-}1\}$ (step length $\Delta = 5$ minutes). Service $s$
has a time-varying resource demand $(c_s(t), m_s(t))$ in vCPU and GiB,
drawn from a real trace; a deploy-time \emph{reservation}
$(R^c_s, R^m_s)$ set to the profiled 99th percentile of its demand
(mirroring conservative production request-setting
practice~\cite{rzadca2020autopilot}); a migration state size $g_s$ in GB;
and an SLO tier $\pi_s \in \{\text{premium}, \text{standard}\}$
($\Pr[\text{premium}]=0.3$).

\paragraph{Catalog.} The instance catalog $\mathcal{K}$ contains 24 types
across three providers $P=\{\text{aws},\text{azure},\text{gcp}\}$, each
with capacity $(C_k, M_k)$, an on-demand price $p^{od}_k$ \$/h, and where
offered, a spot price $p^{sp}_k$ with a per-step interruption hazard
$h_k$. Hazards derive from mean lifetimes: for AWS we map the Spot Instance
Advisor's published per-type interruption-frequency
buckets~\cite{awsspotadvisor} to mean lifetimes (30, 15, 7, 3, 1 days for
the five buckets); for Azure and GCP, which publish no equivalent
advisory, we \emph{assume} 7-day and 5-day mean lifetimes respectively and
sweep this assumption by up to $10\times$ in
Sections~\ref{sec:results-guard} and~\ref{sec:results-sens}. An
interruption kills the instance and every service on it. Cross-provider
migration of service $s$ incurs egress cost $g_s \cdot e_{P(src)}$ with
$e = (\$0.09, \$0.087, \$0.12)$ per GB for AWS, Azure, GCP respectively,
and 2 steps of downtime for the moved service; instance boot takes 1 step;
interrupted services restart after 1 step plus boot.

\paragraph{Objective.} A service-step is \emph{violated} if the service is
down (pending, booting, migrating, or restarting) or its host's actual
demand exceeds capacity in either dimension. Writing $V^{std}$ and
$V^{prem}$ for violated standard and premium service-steps and
$\mathrm{Cost}$ for the sum of on-demand, spot, and egress charges over
the episode, the penalized cost is
\begin{equation}
J \;=\; \mathrm{Cost} \;+\; \beta_{std}\, V^{std} \;+\;
\beta_{prem}\, V^{prem},
\label{eq:J}
\end{equation}
with $\beta_{std} = \$0.5$ and $\beta_{prem} = \$5$.
The $10\times$ asymmetry encodes SLO-credit asymmetry between tiers; it is
the single scalarization used for evolution fitness and headline
comparisons, and we always report its components (cost, per-tier violation
rates) separately so readers can re-weigh them.

\section{The bottleneck: LLMs cannot sit on the decision path}
\label{sec:bottleneck}

Consider the naive architecture: at each step, for each service requiring
action, serialize the fleet context and candidate set into a prompt and let
an LLM choose. The measured decision context in our system (service
features, $\sim$70 candidate descriptors, fleet aggregates) serializes to
3--6\,k tokens. With $N{=}200$ services re-evaluated every 5 minutes, that
is $200 \times 288 = 57{,}600$ decision instances per day, i.e.\
$\sim 2\text{--}3.5 \times 10^{8}$ input tokens daily---at current
frontier-model prices on the order of \$500--1{,}000 per day of inference
for a fleet whose entire measured capacity bill is \$240--520 per day
(Table~\ref{tab:main}). Latency compounds the problem: a single LLM call's
seconds-scale tail would dominate a control loop whose complete per-step
decision cost we measure at 1.9\,ms for $N{=}200$
(Table~\ref{tab:overhead}), and interposing it on interruption-recovery
paths would directly extend premium downtime. Batching and distillation
shrink but do not remove the gap, and they surrender the transparency of
an auditable policy.

The conclusion is architectural: the LLM's contribution must be
\emph{amortized}. HELIOS therefore spends its entire LLM budget offline---%
in our runs, 5 synthesis rounds---producing a $\sim$40-line auditable
program whose marginal decision cost is microseconds. The same amortization
argument underlies FunSearch-style
discovery~\cite{romeraparedes2024funsearch}; our contribution is carrying
it into a safety-constrained online control setting.

\section{HELIOS design}\label{sec:design}

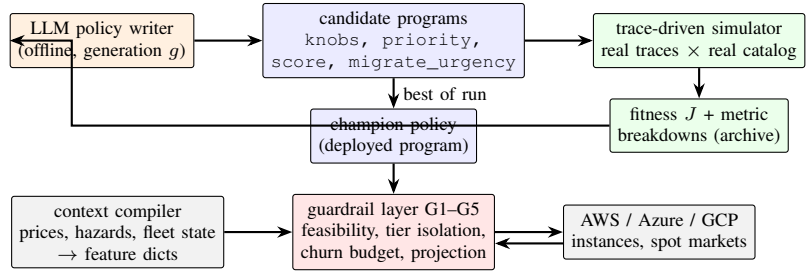
\begin{figure*}[t]
\centering
\begin{tikzpicture}[font=\scriptsize, node distance=4mm and 6mm,
  box/.style={draw, rounded corners=1pt, align=center, inner sep=3pt,
              minimum height=6mm},
  arr/.style={-{Stealth[length=2mm]}, thick}]
\node[box, fill=orange!12] (llm) {LLM policy writer\\(offline, generation $g$)};
\node[box, right=9mm of llm, fill=blue!8] (cands) {candidate programs\\\texttt{knobs, priority,}\\\texttt{score, migrate\_urgency}};
\node[box, right=9mm of cands, fill=green!8] (sim) {trace-driven simulator\\real traces $\times$ real catalog};
\node[box, below=of sim, fill=green!8] (fit) {fitness $J$ + metric\\breakdowns (archive)};
\node[box, below=of cands, fill=blue!8] (champ) {champion policy\\(deployed program)};
\node[box, below=of champ, fill=red!10] (guard) {guardrail layer G1--G5\\feasibility, tier isolation,\\churn budget, projection};
\node[box, left=9mm of guard, fill=gray!10] (ctx) {context compiler\\prices, hazards, fleet state\\$\to$ feature dicts};
\node[box, right=9mm of guard, fill=gray!10] (clouds) {AWS / Azure / GCP\\instances, spot markets};
\draw[arr] (llm) -- (cands);
\draw[arr] (cands) -- (sim);
\draw[arr] (sim) -- (fit);
\draw[arr] (fit.west) -| ([xshift=-4mm]llm.south) |- (llm.west);
\draw[arr] (cands) -- node[right]{best of run} (champ);
\draw[arr] (champ) -- (guard);
\draw[arr] (ctx) -- (guard);
\draw[arr] (guard) -- (clouds);
\draw[arr] ([yshift=-1.5mm]clouds.west) -- ([yshift=-1.5mm]guard.east);
\end{tikzpicture}
\caption{HELIOS architecture. The LLM operates only in the offline
evolution loop (top); the deployed artifact is the champion program running
inside the guardrail layer (bottom), fed by a context compiler that reduces
raw cloud state to the fixed feature interface.}
\label{fig:arch}
\end{figure*}

Figure~\ref{fig:arch} shows the system. Three components matter: the policy
interface that the LLM programs against, the evolution loop that selects
programs by measured fitness, and the guardrail layer that constrains what
any program can do.

\subsection{Policy interface and context compiler}\label{sec:interface}

A policy is a Python class with four pure functions:
\texttt{knobs(ctx)} returns global knobs (currently a placement headroom
fraction in $[0,0.6]$); \texttt{priority(sv, ctx)} orders pending services
for placement within a step; \texttt{score(sv, cand, ctx)} ranks feasible
placement targets; and \texttt{migrate\_urgency(sv, host, ctx)} proposes
(positive return) and prioritizes migrations. The \emph{context compiler}
reduces raw state to three feature dictionaries: per-service \texttt{sv}
(demand, tier, state size, recent trend and peak, residency), per-candidate
\texttt{cand} (provider, market, \$/vCPU-h at current prices, hazard, free
capacity, utilization after placement, boot penalty, egress cost of the
move), and fleet context \texttt{ctx} (aggregate demand, time of day,
pending count, per-provider price minima---which also encode provider
availability during outages). Policies see engineered features, never raw
APIs; this keeps programs short, auditable, and portable between the
simulator and a production executor.

\subsection{LLM-driven evolution}\label{sec:evolution}

\begin{algorithm}[t]
\caption{HELIOS policy evolution (as executed; archive in artifact)}
\label{alg:evolve}
\begin{algorithmic}[1]
\State $\mathcal{A} \gets \{\text{seed heuristic}\}$; \quad
$\mathrm{prompt}_0 \gets$ interface spec + catalog facts + guardrails
\For{generation $g = 0, 1, 2, \dots$}
  \State LLM reads $\mathrm{prompt}_g$ (and for $g{>}0$: full fitness
         table and metric breakdowns of $\mathcal{A}$)
  \State LLM writes $3$--$4$ candidate programs (design hypotheses,
         mutations, crossovers), each with a stated rationale
  \For{each candidate $\pi$}
     \State evaluate $J(\pi)$: 2-day train trace, $N{=}200$,
            seeds $\{0,1\}$; log all metrics
  \EndFor
  \State $\mathcal{A} \gets \mathcal{A} \cup$ candidates;\;
         rank by mean $J$
  \If{no improvement for 2 consecutive generations} \textbf{stop} \EndIf
\EndFor
\State \Return champion $\arg\min_{\pi \in \mathcal{A}} J(\pi)$
\end{algorithmic}
\end{algorithm}

Algorithm~\ref{alg:evolve} is deliberately simple---the interesting
machinery is in what the LLM receives and what it is allowed to produce.
The generation-0 prompt (Appendix~\ref{app:prompt}) contains the interface,
the guardrail contract, and a factual digest of the real catalog (price
minima, discount structure, hazard-to-lifetime mapping, egress rates); it
contains \emph{no} fitness numbers, so generation 0 is pure prior-driven
design. From generation 1 on, the LLM additionally receives the measured
fitness table with cost decompositions, violation rates, migration and
interruption counts for every prior candidate---and nothing else: the LLM
never sees or invents metric values, and every number in its archived
analyses (Appendix references; \texttt{synthesis/} in the artifact) is
traceable to the fitness log. Training uses 2 PlanetLab days and seeds
$\{0,1\}$; all other days, seeds, and the Azure and Bitbrains datasets are
held out.

In this artifact run the policy-writing LLM is Claude (Sonnet-class)
operating in-session; each generation's exact input context and full
response are archived verbatim. We report the evolution trajectory as a
measured result in Section~\ref{sec:results-evo}.

\subsection{Guardrail layer and safety contract}\label{sec:guard}

The guardrail layer intercepts every proposal from the policy program and
enforces:

\begin{description}
\item[G1 (reservation-feasible packing).] A service may be placed on an
instance only if the instance's \emph{reservation} headroom covers
$\max(\text{demand}, R_s)$ in both dimensions; in-flight (migrating or
booting) services hold their reservations at the target. The candidate
generator applies the policy's requested headroom first and relaxes to the
hard reservation check only if no candidate survives (projection, G5).
\item[G2 (tier isolation).] Premium services are never offered spot
candidates: the market filter runs \emph{before} scoring, so no score can
override it.
\item[G3 (churn budget).] At most $\lceil \kappa N \rceil$ migrations
execute per step ($\kappa = 0.05$), highest urgency first.
\item[G4 (residency and cooldown).] A service must be resident $\geq 6$
steps before proposing a move and cools down 6 steps after one.
\item[G5 (projection).] Proposals that violate G1--G4, reference dropped
providers, or exceed the per-step instance-creation cap are dropped or
projected to the nearest feasible action; policy exceptions are contained
and treated as ``no proposal.''
\end{description}

\begin{lemma}[Safety contract]\label{lem:guard}
For every policy program $\pi$ (including adversarial or crashing ones)
and every step $t$: (i) no premium service is ever assigned to spot
capacity; (ii) at most $\lceil \kappa N \rceil$ migrations execute at $t$;
(iii) every assignment performed at $t$ satisfies reservation feasibility
at assignment time; and (iv) the control loop completes step $t$.
\end{lemma}

\begin{IEEEproof}[Proof sketch]
(i) G2 removes spot candidates for premium services from the candidate set
before $\pi$'s \texttt{score} is consulted; $\pi$ can only order the
pre-filtered set. (ii) Migrations execute from a list sorted by urgency and
truncated at the budget; the counter is maintained by the executor, not
$\pi$. (iii) Both passes of the candidate generator check reservation
feasibility against a capacity view in which placed-but-in-flight services
hold their reservations; a candidate failing the check is never surfaced.
(iv) Each call into $\pi$ is exception-guarded; a raised exception yields
an empty proposal set, and the remainder of the step executes mechanisms
only.
\end{IEEEproof}

The contract's value is empirical as well as formal: it converts
catastrophic policy bugs into finite, measurable fitness penalties
(Section~\ref{sec:results-evo}), which is precisely what allows an LLM to
search this space safely, and it bounds behavior under distribution shift
(Section~\ref{sec:results-guard}).

\subsection{The evolved champion}\label{sec:champion}

The final champion (lineage: gen-0 \emph{segregator} $\to$ gen-1
\emph{bigbox} $\to$ gen-2 \emph{amortized}; full code in
Appendix~\ref{app:champion}) is compact enough to audit in one sitting.
Premium services are placed first (priority $10\pi_s + c_s$) and anchored
to the globally cheapest on-demand pool
($\text{score} = -p - 0.5\,|p - p^{od}_{\min}|$). Standard services price
candidates by an \emph{amortized interruption cost},
$p + A \cdot h_k$ with $A{=}10$, derived in generation 1--2 by correcting
a per-service/per-instance accounting error the LLM itself made and then
diagnosed from fitness data. New standard capacity breaks near-price ties
toward larger instances ($+0.0008 \cdot \min(\text{free vCPU}, 12)$),
which the log shows cut interruption events from 18 to 8 per training
episode by shrinking the exposed instance count. Small friction terms
penalize boot and egress, and a single migration rule drifts services off
spot capacity whose amortized price exceeds the cheapest on-demand---the
rule that lets the same constants respond correctly when hazards shift
(Section~\ref{sec:results-sens}).

\section{Implementation}\label{sec:impl}

\paragraph{Simulator.} A 1{,}100-line trace-driven simulator advances the
fleet in 5-minute steps: spot interruptions sample per-instance Bernoulli
hazards; interrupted and migrated services incur the downtime model of
Section~\ref{sec:problem}; instances idle for 2 steps are retired; billing
accrues per step at catalog prices. The candidate generator, guardrails,
and metrics live in the simulator core; policies are plug-in classes.
One 3-day, $N{=}200$ episode executes in $\approx$3\,s on one CPU core,
which is what makes evolution-by-measurement affordable.

\paragraph{Data.} All workload and price data are real and retrieved on
2026-07-03 (full provenance file in the artifact):
PlanetLab CPU traces (10 days, 2011) as packaged with
CloudSim~\cite{calheiros2011cloudsim,park2006comon,beloglazov2012optimal};
the Bitbrains fastStorage traces (CPU and memory,
GWA-T-12~\cite{shen2015bitbrains,gwarchive}); a 1{,}177-VM subset of the
Azure Public Dataset V2~\cite{cortez2017resourcecentral,azuretrace};
cross-provider on-demand and spot prices from the SkyPilot service
catalog~\cite{skypilotcatalog}; and AWS per-type interruption-frequency
buckets from the Spot Instance Advisor~\cite{awsspotadvisor}. Services are
sampled from trace VMs with seeded RNG; PlanetLab and Azure traces are
CPU-only, so memory demand there is synthesized at a fixed 2--4 GiB/vCPU
ratio (disclosed limitation); Bitbrains provides real memory and therefore
serves as an out-of-family generalization test.

\paragraph{Baseline implementations.} The MILP planner solves a
pool-planning integer program with CBC (instance counts $n_{k,m}$, fluid
service assignment $y_{s,k,m}$, reservation capacity constraints with 10\%
headroom, spot-risk term; 2\% optimality gap), receives the
\emph{reservation oracle}---exactly the quantities that determine packing
feasibility---and re-plans whenever the price/availability context
changes. Its executor steers toward the plan with a dominant affinity
bonus. Fluid assignments are restricted to per-instance-feasible
(service, type) pairs; without this restriction the executor ping-pongs
services the LP fractionally split across instances. The DQN
baseline~\cite{mnih2015dqn} (numpy implementation: MLP 14-64-64-18, replay,
target network) re-decides every 6 steps over the \emph{same} knob space
evolution searches---spot aversion $A \in \{5, 25, 10^6\}$, headroom
$\in \{0.02, 0.08, 0.2\}$, big-box on/off over the champion skeleton---so
the comparison isolates the search method, not the hypothesis space, and
the action space contains the champion's static configuration as one arm.
It trains for 100 episodes ($9{,}600$ decisions, $2.6\times$ the
38-episode budget evolution consumed) on the same training days.

\section{Experimental setup}\label{sec:setup}

\begin{table}[t]
\centering
\caption{Multi-cloud catalog summary (retrieved 2026-07-03; 8 types per
provider; regions AWS us-east-1, Azure eastus, GCP us-central1). ``Assumed''
lifetimes are disclosed assumptions swept in
Sections~\ref{sec:results-guard}--\ref{sec:results-sens}.}
\label{tab:catalog}
\scriptsize
\setlength{\tabcolsep}{4pt}
\begin{tabular}{lccc}
\toprule
 & AWS & Azure & GCP\\
\midrule
on-demand \$/vCPU-h & 0.0425 & 0.0423 & 0.0486\\
spot discount range & 43--72\% & 89\% & 73--75\%\\
interruption data & per-type & assumed & assumed\\
mean spot lifetime & 1--30 d & 7 d & 5 d\\
egress \$/GB & 0.09 & 0.087 & 0.12\\
\bottomrule
\end{tabular}
\end{table}

\paragraph{Trace suites.} Training used PlanetLab days 20110303+06 with
seeds $\{0,1\}$ only. Evaluation uses (i) three held-out PlanetLab days
(20110309, 20110322, 20110325) as one 3-day episode, (ii) one day of the
Azure trace, (iii) three days of Bitbrains (real memory dimension), all
with unseen seeds $\{2,3,4\}$, $N{=}200$ (57{,}600--172{,}800
service-steps per episode).

\paragraph{Baselines.} \emph{SingleCloud-BFD}: best-fit-decreasing
consolidation on the cheapest single provider, on-demand only (the
``lift-and-shift'' reference). \emph{Threshold-HPA}: utilization-threshold
scaling in the style of production horizontal autoscalers, on-demand.
\emph{Greedy-MultiCloud}: cheapest risk-adjusted price
$p + 50 h_k$ across all providers and markets---a strong calibrated
heuristic. \emph{Seed-RiskAware}: the hand-written generation-0 seed
(risk-adjusted price, mild consolidation, egress-aware migration).
\emph{MILP-Plan} and \emph{DQN-Meta} as above. All baselines run inside
the same guardrails.

\paragraph{Metrics.} Total cost with od/spot/egress decomposition,
overall and premium violation rates, migrations, interruptions, and
$J$~(Eq.~\ref{eq:J}). In tables, $J$ is recomputed from per-tier violation
rates assuming the nominal 30\% premium share; the reconstruction error
versus exact per-seed premium counts is below 1\%.

\section{Results}\label{sec:results}

\subsection{Evolution dynamics}\label{sec:results-evo}

\begin{figure}[t]
\centering
\includegraphics[width=\columnwidth]{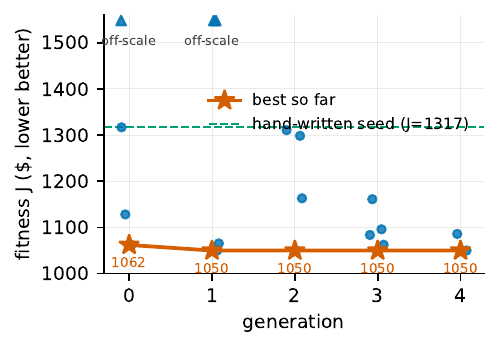}
\caption{Evolution on the training fitness (2 PlanetLab days, seeds
\{0,1\}). Dots are candidate programs; triangles mark candidates off the
scale (up to $J{=}1.4\times 10^5$); the line tracks the best-so-far.
Fitness improves 20.9\% over the hand-written seed and plateaus after
generation 1.}
\label{fig:evolution}
\end{figure}

Figure~\ref{fig:evolution} plots all 19 candidates
(Appendix~\ref{app:evtable} lists every fitness). Three findings shape the
rest of the paper.

\emph{(1) The gains are structural, not parametric.} The two moves that
matter---segregating tiers with premium-first placement ordering
(generation 0, worth more than \$77 of capacity purely through 5$\times$
-weighted premium cold-start ordering) and big-box spawning (generation 1,
interruptions 18$\to$8 by shrinking exposed instance count)---are design
decisions. Generations 3--4 tested seven parameter-level refinements of the
champion; every one measurably regressed (densifying packing or reordering
placement leaked 5$\times$-weighted premium violation steps), and two bolder
gene proposals were decision-identical to the champion on the training
trace. Evolution stopped by the two-generation patience rule.

\emph{(2) Negative results are selected against, cheaply.} The
actual-utilization drain gene---consolidate by migrating services off
lightly-loaded hosts, a standard idea in single-DC
consolidation~\cite{beloglazov2012optimal}---failed catastrophically twice
(generation 0: $J{=}142{,}137$, 2{,}052 migrations, 61\% violations;
a ``repaired'' generation-1 variant: 775 migrations, 29\%). The mechanism
is instructive: under p99 reservations, \emph{actual} utilization sits near
0.15 fleet-wide, so the trigger fires everywhere and drained services land
on hosts about to look drainable themselves. The guardrails (churn budget,
residency) kept these episodes finite and measurable rather than
divergent---exactly the property Lemma~\ref{lem:guard} promises---and the
fitness signal killed the gene family. Similarly, removing the premium
min-price anchor looked like a simplification and cost $+\$221$ of
on-demand spend: the big-box tiebreak coefficient exceeded the on-demand
\$/vCPU spread and spawned fragmenting 16-vCPU boxes. The transferable
rule the LLM articulated from the data---\emph{tiebreak coefficients must
sit below the price spread of the pool they act on}---is the kind of
design knowledge this loop makes explicit and archives.

\emph{(3) The LLM's own arithmetic errors are recoverable.} The
generation-1 amortized-risk candidate over-penalized hazards by
$\sim$20$\times$ through a per-service versus per-instance accounting
error, driving the fleet all-on-demand (\$984). The generation-2 analysis
diagnosed the error from the cost decomposition and corrected the constant
($A{=}216 \to 10$); the corrected form is decision-identical to the
hard-ceiling champion on calibrated hazards but degrades gracefully when
hazards shift, which is why it carries the final policy.

Under the corrected simulator accounting (an in-flight reservation bug we
found and fixed during evaluation), we re-ran all 5 generations: every
per-generation ranking is unchanged and the champion's fitness trajectory
is identical (the champion never migrates in training), so champion
selection is robust to the fix; both logs ship in the artifact.

\subsection{Main comparison}\label{sec:results-main}

\begin{table*}[t]
\centering
\caption{Main results (mean $\pm$ std over unseen seeds 2--4; all
policies inside identical guardrails). Viol.\ columns are the percentage
of violated service-steps overall and for the premium tier. $J$ per
Eq.~\ref{eq:J}; best per suite in bold.}
\label{tab:main}
\scriptsize
\setlength{\tabcolsep}{3.6pt}
\begin{tabular}{lcccccc}
\toprule
policy & cost (\$) & intr. & migr. & viol.\ \% & prem.\ \% & $J$ (\$)\\
\midrule
\multicolumn{7}{l}{\emph{PlanetLab held-out (3 days)}}\\
SingleCloud-BFD & {\$1552} $\pm$ 133 & 0 & 0 & 0.249 & 0.237 & {2320}\\
Threshold-HPA & {\$1520} $\pm$ 78 & 0 & 14 & 0.260 & 0.260 & {2351}\\
Greedy-MultiCloud & {\$665} $\pm$ 49 & 12 & 0 & 0.257 & 0.228 & {1419}\\
Seed-RiskAware & {\$823} $\pm$ 64 & 7 & 13 & 0.250 & 0.235 & {1586}\\
MILP-Plan & {\$1207} $\pm$ 203 & 13 & 0 & 0.216 & 0.133 & {1704}\\
DQN-Meta & {\$1575} $\pm$ 164 & 0 & 0 & 0.297 & 0.212 & {2326}\\
\textsc{Helios} & \textbf{\$723} $\pm$ 26 & 11 & 0 & 0.215 & 0.161 & \textbf{1283}\\
\addlinespace
\multicolumn{7}{l}{\emph{Azure 2019 (1 day)}}\\
SingleCloud-BFD & {\$431} $\pm$ 17 & 0 & 0 & 0.549 & 0.531 & {1002}\\
Threshold-HPA & {\$429} $\pm$ 32 & 0 & 6 & 0.596 & 0.571 & {1044}\\
Greedy-MultiCloud & {\$216} $\pm$ 15 & 3 & 0 & 0.581 & 0.607 & {856}\\
Seed-RiskAware & {\$250} $\pm$ 10 & 4 & 3 & 0.609 & 0.549 & {852}\\
MILP-Plan & {\$304} $\pm$ 68 & 3 & 0 & 0.616 & 0.381 & {778}\\
DQN-Meta & {\$481} $\pm$ 34 & 0 & 0 & 0.771 & 0.540 & {1123}\\
\textsc{Helios} & \textbf{\$231} $\pm$ 8 & 3 & 0 & 0.649 & 0.457 & \textbf{773}\\
\addlinespace
\multicolumn{7}{l}{\emph{Bitbrains fastStorage (3 days)}}\\
SingleCloud-BFD & {\$1246} $\pm$ 81 & 0 & 0 & 0.216 & 0.231 & {1971}\\
Threshold-HPA & {\$1272} $\pm$ 75 & 0 & 14 & 0.211 & 0.215 & {1954}\\
Greedy-MultiCloud & {\$569} $\pm$ 58 & 10 & 0 & 0.244 & 0.242 & {1345}\\
Seed-RiskAware & {\$662} $\pm$ 71 & 6 & 18 & 0.223 & 0.251 & {1440}\\
MILP-Plan & {\$1737} $\pm$ 183 & 13 & 0 & 0.197 & 0.116 & {2177}\\
DQN-Meta & {\$1120} $\pm$ 61 & 0 & 0 & 0.196 & 0.176 & {1701}\\
\textsc{Helios} & \textbf{\$775} $\pm$ 62 & 11 & 0 & 0.203 & 0.136 & \textbf{1268}\\
\bottomrule

\end{tabular}
\end{table*}

\begin{figure*}[t]
\centering
\includegraphics[width=\textwidth]{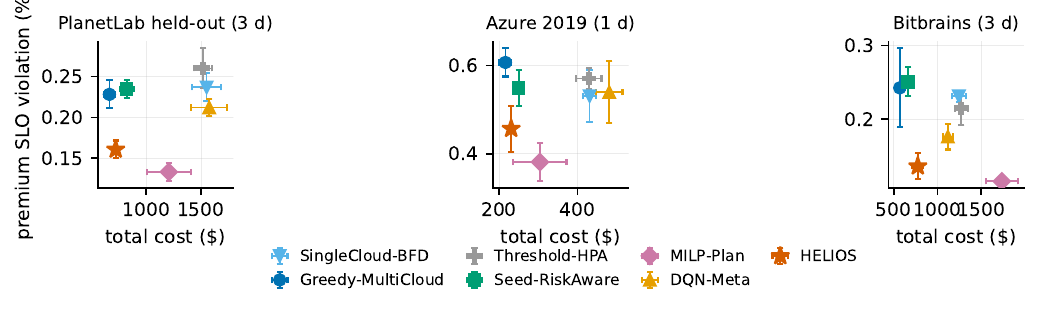}
\caption{Cost versus premium SLO violation on the three held-out trace
suites (mean $\pm$ std over seeds). \textsc{Helios} (star) sits on or near
the Pareto frontier of every suite.}
\label{fig:pareto}
\end{figure*}

Table~\ref{tab:main} and Figure~\ref{fig:pareto} give the headline result:
\textsc{Helios} attains the lowest penalized cost on \emph{all three}
held-out suites---$-45$\% versus SingleCloud-BFD and $-9.6$\% versus the
calibrated Greedy on PlanetLab; $-$9.6\% versus Greedy on Azure; $-5.7$\%
on Bitbrains---and does so from different directions. Against the
single-cloud and autoscaler baselines the win is raw cost (spot capacity
plus provider choice, cost $-53$\% versus BFD). Against Greedy the win is
SLO: Greedy is 8\% cheaper on PlanetLab but incurs 42\% more premium
violation steps (0.228\% vs 0.161\%), which the tiered objective correctly
prices; at $N{=}400$ \textsc{Helios} is cheaper in raw dollars as well
(\$1{,}273 vs \$1{,}400, Section~\ref{sec:results-sens}) as the big-box
gene compounds. Against the hand-written seed---the same designer's best
manual attempt in the same interface---evolution is worth $-12$\% cost
\emph{and} $-31$\% premium violations simultaneously.

Two baselines deserve their own paragraphs. \emph{MILP-Plan}, despite the
reservation oracle and an exact solver, costs 40--67\% more than
\textsc{Helios} on the CPU-only suites and 124\% more on Bitbrains: integer
rounding across many small per-type pools plus uniform planning headroom
over-provisions, and heterogeneous real memory reservations aggravate the
rounding. It does deliver the best premium SLO on PlanetLab and Bitbrains
(0.133\%/0.116\%)---the classic planning trade of dollars for
predictability---but Section~\ref{sec:results-robust} shows the plan's
brittleness under disturbance. \emph{DQN-Meta} fails outright: it
converges to the maximally risk-averse arm (all on-demand, zero spot),
costing more than BFD because it also over-provisions instances. The
diagnosis from training curves is credit assignment: exploratory arm
switches mid-episode trigger migration waves whose delayed violation
penalties swamp the per-epoch reward signal, teaching the agent that
\emph{any} spot exposure is dangerous. Under a comparable (indeed
$2.6\times$ larger) interaction budget, evolving a static program was both
easier to search and better at the optimum than learning a time-varying
controller over the same knobs---a finding consistent with the
sample-efficiency critiques of DRL-for-systems~\cite{mao2019decima} and,
we believe, of independent interest.

Generalization deserves emphasis: evolution saw only 2 PlanetLab CPU-only
days, yet the ranking holds on a different cloud's workload (Azure) and on
Bitbrains, whose \emph{real memory dimension} the training distribution
did not contain. The champion's genes (tier segregation, exposure
minimization, amortized hazard pricing) are properties of the
\emph{catalog structure}, not the trace, which is the mechanistic reason
they transfer.

\subsection{Ablations}\label{sec:results-abl}

\begin{table*}[t]
\centering
\caption{Component ablations on PlanetLab held-out (seeds 2--4).
$^{\dagger}$At calibrated hazards; see Table~\ref{tab:guard} for why
$-$guardrails is not a viable configuration.}
\label{tab:ablation}
\scriptsize
\begin{tabular}{lccccc}
\toprule
variant & cost & $\Delta$cost & viol.\% & prem.\% & $\Delta J$\\
\midrule
\textsc{Helios} (full) & \$723 & +0\% & 0.215 & 0.161 & +0\%\\
$-$\,evolution (seed heuristic) & \$823 & +14\% & 0.250 & 0.235 & +24\%\\
$-$\,spot markets & \$1732 & +140\% & 0.274 & 0.199 & +90\%\\
$-$\,migration & \$723 & +0\% & 0.215 & 0.161 & +0\%\\
single cloud (Azure) & \$761 & +5\% & 0.213 & 0.156 & +2\%\\
single cloud (AWS) & \$1521 & +111\% & 0.312 & 0.168 & +70\%\\
single cloud (GCP) & \$1934 & +168\% & 0.236 & 0.116 & +88\%\\
$-$\,guardrails$^{\dagger}$ & \$218 & -70\% & 0.206 & 0.182 & -36\%\\
\bottomrule

\end{tabular}
\end{table*}

\begin{figure}[t]
\centering
\includegraphics[width=\columnwidth]{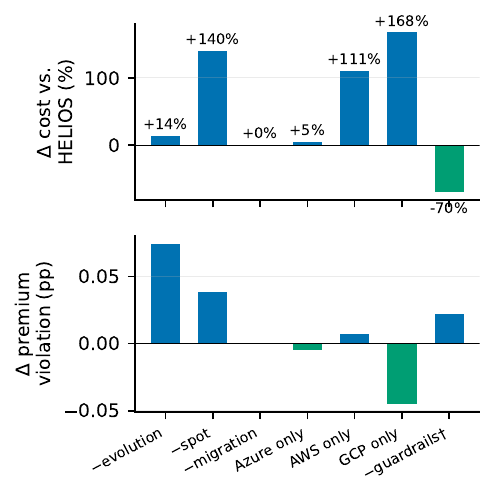}
\caption{Ablation deltas versus full \textsc{Helios} (PlanetLab held-out).}
\label{fig:ablation}
\end{figure}

Table~\ref{tab:ablation} and Figure~\ref{fig:ablation} decompose the win.
Spot capacity is the largest single factor ($-$spot: $+140$\% cost).
Provider choice is the second: the fleet confined to AWS or GCP costs
$+111$\% and $+168$\%, while Azure-only costs just $+5$\%---because
Azure's spot pricing happens to dominate this 2026 snapshot, multi-cloud's
\emph{steady-state} value here is mostly automatic best-provider
selection. Its \emph{contingent} value is what the robustness suite prices
(Section~\ref{sec:results-robust}): the Azure-only variant has no
provider-outage fallback at any cost. Evolution contributes $+14$\% cost /
$+24$\% $J$ when removed. Migration contributes nothing in calm held-out
conditions---the champion executes zero migrations there, and the
$-$migration row is identical---but becomes load-bearing under
disturbances, where the drift-off-overpriced-spot rule and forced
re-placements do the work. We keep the gene on the strength of the
robustness results rather than the calm-weather ablation, and flag this
honestly.

\subsection{The price and value of guardrails}\label{sec:results-guard}

\begin{table}[t]
\centering
\caption{Guardrail value under spot-hazard stress (PlanetLab held-out,
seeds 2--3). Each cell: total cost / premium violation rate.}
\label{tab:guard}
\begin{tabular}{lcccc}
\toprule
 & \multicolumn{2}{c}{\textsc{Helios} (guarded)} &
   \multicolumn{2}{c}{no guardrails}\\
hazard & cost & prem.\ viol. & cost & prem.\ viol.\\
\midrule
$\times$1 & \$731 & 0.17\% & \$219 & 0.18\%\\
$\times$4 & \$691 & 0.17\% & \$146 & 0.16\%\\
$\times$10 & \$1412 & 0.17\% & \$825 & 97.68\%\\
\bottomrule

\end{tabular}
\end{table}

\begin{figure*}[t]
\centering
\includegraphics[width=\textwidth]{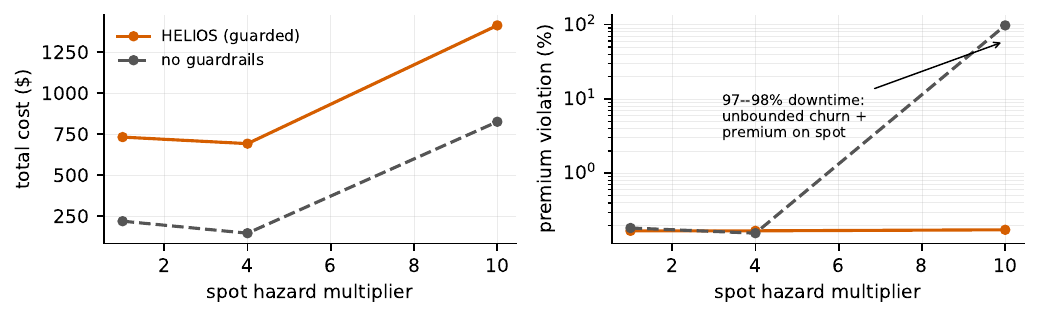}
\caption{Cost (left) and premium violation rate (right, log scale) as the
spot hazard scales. Unguarded operation is cheap until it is
catastrophic.}
\label{fig:guardvalue}
\end{figure*}

Removing the guardrails at calibrated hazards is, superficially,
free money: premium services flee to 89\%-off spot and the bill drops 70\%
with violation rates statistically unchanged (Table~\ref{tab:ablation}),
because at 7-day mean lifetimes the \emph{mean} interruption rate barely
moves the SLO needle. This is precisely why such constraints get argued
away in practice. Table~\ref{tab:guard} and Figure~\ref{fig:guardvalue}
supply the counter-argument quantitatively: scale hazards $10\times$
(sub-day mean lifetimes---the regime of a real capacity crunch, and within
the historical range of AWS's own published
buckets~\cite{awsspotadvisor}) and the unguarded system collapses to
97--98\% premium downtime through a compounding failure: premium services
sit on evaporating spot capacity while the policy's own migration rule,
now unthrottled (no churn budget, no residency), thrashes the fleet.
Guarded \textsc{Helios} in the same regime holds premium violations at
0.17--0.18\% and pays a graceful $+100$\% cost as the amortized price
correctly pushes standard services back to on-demand. Guardrails are not
overhead; they are the difference between a risk premium and a tail event.
An incidental finding worth reporting: at intermediate hazards
($\times 4$) cost \emph{decreases} for both configurations---interruption
churn acts as free defragmentation, forcing periodic re-packing into
existing capacity.

\subsection{Sensitivity}\label{sec:results-sens}

\begin{figure*}[t]
\centering
\includegraphics[width=\textwidth]{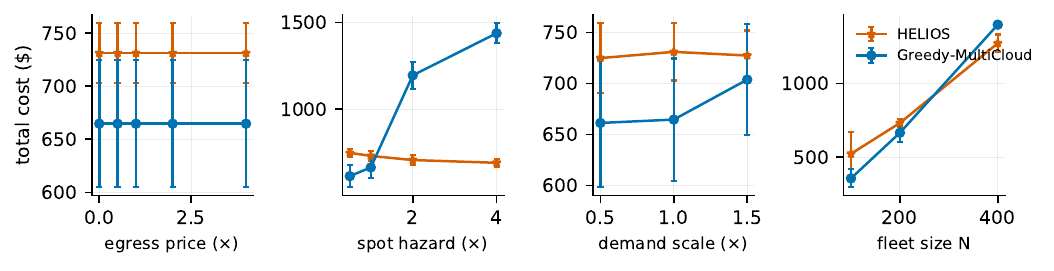}
\caption{Cost sensitivity (PlanetLab held-out, seeds 2--3):
egress price, spot hazard, demand scale, and fleet size.}
\label{fig:sensitivity}
\end{figure*}

Figure~\ref{fig:sensitivity} sweeps four environment parameters against
the strongest cost competitor. Both policies are egress-insensitive in
steady state (neither migrates across providers absent disturbances).
The hazard sweep separates them sharply: Greedy's fixed risk weight
($50 h$) abandons spot near $\times 2$ and its cost more than doubles by
$\times 4$ (\$1{,}445 vs \textsc{Helios} \$691), while the champion's
amortized constant ($10 h$), derived from the actual downtime penalty an
interruption imposes, keeps cheap spot correctly priced until far larger
shifts. Demand scaling is benign for both. Fleet-size scaling crosses
over: at $N{=}100$ Greedy's simplicity wins, by $N{=}400$
\textsc{Helios}'s exposure-minimizing packing is 9\% cheaper in raw
dollars---the regime where orchestration sophistication pays.

\subsection{Robustness: price shocks and provider outages}
\label{sec:results-robust}

\begin{table}[t]
\centering
\caption{Robustness scenarios (PlanetLab held-out, seeds 2--3), disturbance
at episode midpoint.}
\label{tab:robust}
\begin{tabular}{lcccc}
\toprule
policy & cost & viol.\% & prem.\% & migr.\\
\midrule
\multicolumn{5}{l}{\emph{calm (no disturbance)}}\\
\textsc{Helios} & \$731 & 0.211 & 0.168 & 0\\
Greedy-MultiCloud & \$665 & 0.249 & 0.234 & 0\\
MILP-Plan & \$1084 & 0.206 & 0.137 & 0\\
\addlinespace
\multicolumn{5}{l}{\emph{Azure on-demand price $\times$1.3 at $t{=}1.5$\,d}}\\
\textsc{Helios} & \$805 & 0.211 & 0.168 & 0\\
Greedy-MultiCloud & \$729 & 0.249 & 0.234 & 0\\
MILP-Plan & \$819 & 0.309 & 0.336 & 95\\
\addlinespace
\multicolumn{5}{l}{\emph{Azure zone outage at $t{=}1.5$\,d}}\\
\textsc{Helios} & \$883 & 0.408 & 0.252 & 0\\
Greedy-MultiCloud & \$929 & 0.341 & 0.374 & 0\\
MILP-Plan & \$1214 & 0.372 & 0.253 & 22\\
\bottomrule

\end{tabular}
\end{table}

\begin{figure}[t]
\centering
\includegraphics[width=\columnwidth]{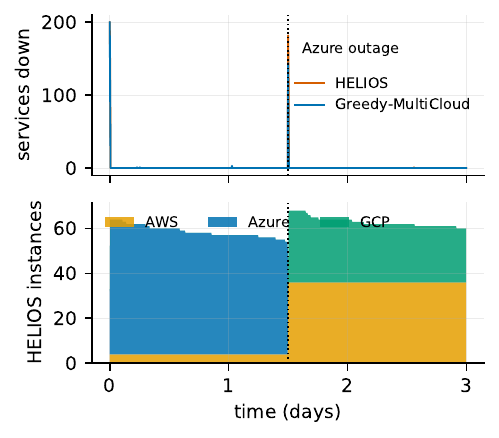}
\caption{Azure zone outage case study (seed 2). Top: services down over
time---the outage spike resolves within a few control periods. Bottom:
\textsc{Helios}'s instance fleet by provider; the Azure share is rebuilt
on GCP and AWS within the same period.}
\label{fig:outage}
\end{figure}

Table~\ref{tab:robust} and Figure~\ref{fig:outage} evaluate the scenarios
multi-cloud exists for. Under a mid-episode Azure zone outage that kills
the majority of the fleet's instances, \textsc{Helios} re-places the
displaced services onto GCP and AWS within a few control periods
(Figure~\ref{fig:outage}) and holds premium violations to 0.25\%---34\%
fewer than Greedy (0.37\%), whose refugees re-place but whose lack of
premium-first ordering and anchoring costs it on the 5$\times$-weighted
tier, at 5\% higher cost. MILP-Plan, given availability-aware re-planning,
also recovers its SLO (0.25\%) but pays 37\% more than \textsc{Helios}
for the recovery: the fresh global plan triggers a coordinated migration
wave toward newly-optimal pools, each move costing downtime, where
\textsc{Helios}'s incremental policy simply re-places refugees and leaves
settled services alone. Under a 30\% Azure on-demand price shock,
\textsc{Helios}'s premium violations are \emph{unchanged from calm
conditions} (its anchor migrates premium services off the shocked pool at
bounded churn), while MILP-Plan's re-plan again converts a price change
into a violation spike (0.34\% premium, $2.6\times$ its calm rate). The
architectural lesson generalizes: under disturbance, incremental
per-service adaptation dominates global re-planning, because a fresh
optimum far from the current assignment is expensive to reach through
migration downtime.

\subsection{Decision overhead}\label{sec:results-overhead}

\begin{table}[t]
\centering
\caption{\textsc{Helios} decision overhead (single CPU core).}
\label{tab:overhead}
\begin{tabular}{ccc}
\toprule
$N$ & decision ms / step & episode wall-clock (s, 3 days)\\
\midrule
100 & 0.98 & 1.2\\
200 & 1.93 & 2.2\\
400 & 4.17 & 4.7\\
\bottomrule

\end{tabular}
\end{table}

Table~\ref{tab:overhead}: the deployed artifact costs $\sim$10\,\textmu s
per service per step and scales linearly---six orders of magnitude below
the per-decision LLM architecture of Section~\ref{sec:bottleneck}, and
comfortably inside a 5-minute (or 5-second) control period.

\section{Discussion and limitations}\label{sec:limits}

\paragraph{Simulation, not production.} All results are trace-driven
simulation. The simulator's downtime, boot, and egress models are simple
constants; real migrations have workload-dependent brownouts, and real
boot times vary by image and provider. We mitigated the classic
sim-optimism failure modes we encountered---and disclose them: an early
accounting bug let concurrently migrating services stack onto one target
(fixed; every experiment re-run; evolution rankings unchanged), and
reservation-based packing was introduced after observing that
demand-based packing produced unrealistic 30\% violation regimes.

\paragraph{Price and hazard snapshots.} The catalog is a single 2026-07-03
snapshot. Azure's 89\% spot discount dominates this snapshot and shrinks
the \emph{steady-state} multi-cloud gap (Azure-only: $+5$\%); under other
snapshots the provider mix, and possibly the champion's constants, would
differ. Azure and GCP mean spot lifetimes are assumptions (7\,d/5\,d),
not measurements; we swept them $10\times$ and the qualitative
conclusions (guardrail value, amortized-price adaptivity) are exactly the
ones that survive the sweep. Intra-episode spot \emph{price} dynamics are
not modeled (post-2017 AWS spot prices move
slowly~\cite{benyehuda2013spot}, but capacity-driven eviction waves do
not; our hazard sweep is a coarse proxy).

\paragraph{Workload modeling.} PlanetLab and Azure traces are CPU-only;
their memory demand is synthesized. The premium tier is assigned randomly
at 30\%, and violation costs (\$0.5/\$5 per 5-minute service-step) are a
stylized SLO-credit model; different weightings would move the
cost--SLO trade-off along Figure~\ref{fig:pareto}'s frontier rather than
reorder its points, but we have not verified this beyond the two weights
reported.

\paragraph{The LLM in the loop.} The policy-writing LLM ran in-session
under authorial supervision rather than as an unattended API pipeline; the
archive (prompts, responses, code, fitness) makes the process auditable
and the harness is fully automated, but replication with other models is
future work, as is scaling the search (our 19-program budget is tiny next
to FunSearch-scale sampling~\cite{romeraparedes2024funsearch}; the fitness
plateau at generation 2 may reflect the budget as much as the landscape).
Evolved constants are snapshot-fitted; continual re-evolution against
rolling price feeds is the natural deployment mode and its stability is
untested.

\paragraph{Scope.} We orchestrate stateless-restartable long-running
services at the placement/migration layer; we do not model network
latency between clouds, data-gravity constraints, GPU capacity, or
multi-region placement within a provider, all of which the interface can
express as features but the present evaluation does not exercise.

\section{Conclusion}\label{sec:conclusion}

HELIOS shows that the right place for an LLM in high-frequency cloud
resource control is off the critical path: as an offline synthesizer of
small, auditable policy programs, selected by measured fitness in a
trace-driven simulator built from real workloads and real prices, and
deployed inside a guardrail layer whose safety contract holds regardless
of what the generated code does. The evolved policy beats every baseline
we could calibrate---including an oracle-informed MILP and a DQN with a
larger interaction budget---on penalized cost across three real trace
suites, transfers to a workload dimension it never trained on, adapts
correctly to hazard and price disturbances, and costs microseconds per
decision. The evolution archive itself is a contribution: it records not
just the champion but the measured failure of plausible designs, turning
policy engineering folklore into data.

\section*{Reproducibility}
The complete artifact---the simulator, all evolved policies, the LLM
prompt/response archive, the experiment logs, and the table-generation
scripts that produce every number in this paper---accompanies the
manuscript.

\bibliographystyle{ieeetr}
\bibliography{refs}

\appendices

\section{Generation-0 synthesis prompt (abridged)}\label{app:prompt}
{\scriptsize The complete prompt and all responses are archived under
\texttt{synthesis/} in the artifact. The prompt comprises: the fitness
definition of Eq.~\ref{eq:J} and training configuration; the policy API of
Section~\ref{sec:interface} with all feature-dictionary fields; the
environment fact digest (provider price minima, spot discount structure,
hazard-to-lifetime mapping, egress rates, interruption/migration downtime
constants); the guardrail contract G1--G5; and the deliverable
specification (``3 diverse policy programs\ldots prioritize different
design hypotheses, not parameter variants''). From generation 1 the prompt
additionally contains the full measured fitness table of all prior
candidates with cost decompositions, violation rates, and
migration/interruption counts.}

\section{Champion policy (complete)}\label{app:champion}
\begin{lstlisting}
"""HELIOS evolved policy (champion of 5-generation LLM-driven synthesis;
lineage: gen0_d_segregator -> gen1_d_bigbox -> gen2_a_amortized_bigbox).

Genes and their measured provenance (see synthesis/ and results/evolution_log.json):
  - premium-first placement ordering        [gen0: -$77-equivalent premium downtime]
  - premium anchored to cheapest on-demand  [gen2: removing it cost +$221]
  - amortized interruption price A=10       [gen1->2: corrected per-instance accounting]
  - bigbox spawn tiebreak (standard pool)   [gen1: interruptions 18 -> 8]
  - fill-existing bonus, boot & egress friction terms [gen0 seed]
"""

class Helios:
    name = "HELIOS"
    A = 10.0

    def knobs(self, ctx):
        return {"headroom": 0.08}

    def priority(self, sv, ctx):
        return 10.0 * sv["premium"] + sv["cpu"]

    def migrate_urgency(self, sv, host, ctx):
        if host["market"] == 1 and \
           host["price_vcpu"] + self.A * host["hazard"] > ctx["min_od_vcpu"]:
            return 0.5
        return 0.0

    def score(self, sv, cand, ctx):
        if sv["premium"] > 0:
            sc = -cand["price_vcpu"] - 0.5 * abs(cand["price_vcpu"] - ctx["min_od_vcpu"])
            if cand["new"] == 0: sc += 0.01
        else:
            sc = -(cand["price_vcpu"] + self.A * cand["hazard"])
            if cand["new"] == 0:
                sc += 0.008
            else:
                sc += 0.0008 * min(cand["free_cpu"], 12.0)
        sc -= 0.004 * cand["boot"]
        sc -= 0.008 * cand["egress"]
        return sc

POLICY = Helios
\end{lstlisting}

\section{Per-generation fitness log}\label{app:evtable}
\begin{table*}[h]
\centering
\caption{All 19 evolved candidates (training fitness; corrected-accounting
re-run; both logs in artifact).}
\scriptsize
\begin{tabular}{llccccc}
\toprule
gen & candidate & $J$ & cost & viol.\% & migr. & intr.\\
\midrule
0 & gen0\_d\_segregator & 1062 & \$468 & 0.40\% & 0 & 18\\
0 & gen0\_b\_spotceiling & 1128 & \$390 & 0.38\% & 0 & 12\\
0 & gen0\_a\_seed & 1317 & \$534 & 0.39\% & 8 & 8\\
0 & gen0\_c\_consolidator & 23167 & \$487 & 7.47\% & 4142 & 5\\
\addlinespace
1 & gen1\_d\_bigbox & 1050 & \$500 & 0.32\% & 0 & 8\\
1 & gen1\_a\_premfirst\_ceiling & 1065 & \$468 & 0.40\% & 0 & 18\\
1 & gen1\_c\_amortized\_risk & 1729 & \$984 & 0.43\% & 0 & 0\\
1 & gen1\_b\_safe\_drain & 2260 & \$454 & 2.50\% & 1231 & 8\\
\addlinespace
2 & gen2\_a\_amortized\_bigbox & 1050 & \$500 & 0.32\% & 0 & 8\\
2 & gen2\_c\_refuge & 1163 & \$486 & 0.36\% & 0 & 8\\
2 & gen2\_d\_full & 1298 & \$725 & 0.34\% & 0 & 8\\
2 & gen2\_b\_prem\_pack & 1311 & \$725 & 0.34\% & 0 & 8\\
\addlinespace
3 & gen3\_b\_spot\_fill & 1062 & \$500 & 0.34\% & 0 & 8\\
3 & gen3\_c\_res\_priority & 1084 & \$500 & 0.32\% & 0 & 8\\
3 & gen3\_a\_prem\_fill & 1096 & \$494 & 0.34\% & 0 & 8\\
3 & gen3\_d\_lean\_headroom & 1161 & \$490 & 0.35\% & 0 & 8\\
\addlinespace
4 & gen4\_b\_conservative\_A25 & 1050 & \$500 & 0.32\% & 0 & 8\\
4 & gen4\_c\_driftback & 1050 & \$500 & 0.32\% & 0 & 8\\
4 & gen4\_a\_bc\_cross & 1086 & \$500 & 0.33\% & 0 & 8\\
\bottomrule

\end{tabular}
\end{table*}

\end{document}